\documentclass[12pt]{article}

\usepackage{authblk}

\usepackage{graphicx}%
\usepackage{multirow}%
\usepackage{amsmath,amssymb,amsfonts}%
\usepackage{amsthm}%
\usepackage{mathrsfs}%
\usepackage[title]{appendix}%
\usepackage{xcolor}%
\usepackage{textcomp}%
\usepackage{manyfoot}%
\usepackage{booktabs}%
\usepackage{algorithm}%

\usepackage{algpseudocode}
\usepackage{listings}%

\usepackage{algorithmicx}

\newcommand{\Fp}{\mathbb{F}_p}

\newcommand{\SM}[2]{\text{SM}(#1,#2)}
\newcommand{\sO}{\tilde{\mathcal{O}}}

\newcommand{\FF}{{\rm I\kern -0.21em F}}
\newcommand{\RR}{{\rm I\kern -0.21em R}}
\newcommand{\NN}{{\rm I\kern -0.21em N}}
\newcommand{\CC}{{\rm \rule{.36em}{0ex}\rule{.08em}{1.55ex}\kern -0.36em C}}
\newcommand{\ZZ}{{\rm Z\kern -0.31em Z}}
\newcommand{\Q}{{\rm \rule{.33em}{0ex}\rule{.08em}{1.52ex}\kern -0.33em Q}}
\newcommand{\Fq}{\mathbb{F}_q}
\newcommand{\Fqx}{\mathbb{F}_q[X; \theta]}

\newtheorem{theorem}{Theorem} 

\newtheorem{proposition}{Proposition}
 
\newtheorem{definition}{Definition}
\newtheorem{remark}{Remark} 
\newtheorem{lemma}{Lemma}

\newtheorem{ex}{Example}

\title{Skew CRT codes and their decoding in poly skew metric.}
\author[1]{Kayodé Epiphane Nouetowa}
\affil[1]{XLIM, France, MATHIS/CRYPTIS and Univ Rennes, CNRS, IRMAR - UMR 6625, Rennes Cedex, France\thanks{kayode-epiphane.nouetowa@unilim.fr}}

\author[2]{Olivier Ruatta}
\affil[2]{XLIM, France, MATHIS/CRYPTIS and Inria Bordeaux - Sud-Ouest Research Center \thanks{olivier.ruatta@unilim.fr}}

\date{}
\begin{document}

	\maketitle

\begin{abstract}
In this paper, we introduce skew CRT codes and the poly-skew metric with associated error model, and we provide a decoding algorithm for this family of codes under the poly-skew metric.
\end{abstract}

\section{Introduction}

In this paper, we introduce a new family of codes, the \emph{skew CRT
codes}, together with a tailored metric, the \emph{poly-skew metric},
for which we provide a decoding algorithm with provably good complexity.

\paragraph{Historical background.}
The history of CRT codes can be traced back to Reed--Solomon
codes~\cite{RS}, although the latter were not originally presented using
the Chinese Remainder Theorem (CRT) formalism. The terminology
\emph{Chinese Remainder Codes} was introduced by Stone~\cite{Stone63}
over the ring of integers; Stone also observed that Reed--Solomon codes
are special instances of CRT codes over polynomial rings. While the idea
of polynomial CRT codes was already implicit in the literature, it was
Shiozaki~\cite{Shiozaki88} who first proposed a decoding algorithm for
polynomial CRT codes based on the extended Euclidean algorithm; this
approach was later generalised and improved by Yu and
Loeliger~\cite{YuLoeliger12}. For integer CRT codes, Li~\cite{Li12}
proposed a syndrome decoding algorithm based on a variant of the extended
Euclidean algorithm. One recurring difficulty in this area is the choice
of the weight or metric with respect to which decoding is performed. A
significant effort has been devoted to taking into account the sizes of
the moduli used to compute the remainders, as is the case for
Reed--Solomon codes under the Hamming metric, which admit a quasi-linear
time decoding algorithm also based on the extended Euclidean algorithm.
Recently, Gaborit, Garnier and Ruatta~\cite{GGR26} proposed a
construction based on the CRT for linearized polynomials over finite
fields, the \emph{$q$-CRT codes}, and provided an efficient decoding
algorithm for the rank and sum-rank metrics. Their construction is
analogous to the one developed here; however, their approach requires
multiplying the encoded message by an auxiliary polynomial in order to
control the minimum distance of the code with respect to the rank or sum-rank metric.

\paragraph{Choice of metric.}
When decoding a family of codes, one always seeks the metric most
naturally adapted to the algebraic structure of the construction. This
principle is illustrated by Gabidulin codes, which are most naturally
considered under the rank metric~\cite{Gabidulin1985,loidreau2005}, and by skew
Reed--Solomon codes, which are considered under the skew
metric~\cite{boucher2020}. While both families could formally be viewed
as codes under the Hamming metric, the latter is not the most
appropriate choice. For polynomial CRT codes, the natural metric is the
weighted Hamming metric. As mentioned above, $q$-CRT codes have been
studied under the rank and sum-rank metrics. However, a subfamily of
$q$-CRT codes coincides with skew Reed--Solomon codes, for which the
appropriate metric is the skew metric, suggesting that the rank metric
may not be the most natural choice for $q$-CRT codes in general. This
observation motivates the introduction of a new metric better suited to
the skew CRT setting.

\paragraph{Contributions.}
The skew CRT codes introduced in this paper are CRT codes constructed
from skew polynomials (Ore polynomials) over a finite field. They
generalise skew Reed--Solomon codes and are closely related to $q$-CRT
codes, but without the auxiliary modification introduced in~\cite{GGR26}
to embed them into the rank metric framework. We introduce the
\emph{poly-skew metric}, which is the natural metric for skew CRT codes
and reduces to the skew metric in the special case of skew Reed--Solomon
codes. We further provide an efficient decoding algorithm for skew CRT
codes under the poly-skew metric, extending existing efficient decoders
for skew Reed--Solomon codes~\cite{boucher2020, boucher2025}.

\paragraph{Organisation of the paper.}
The paper is organised as follows. We first recall the necessary
background on skew polynomials over finite fields and on linear codes,
before presenting the construction of skew CRT codes. We then introduce
the poly-skew metric and describe how to construct error vectors of
prescribed poly-skew weight, thereby establishing the error model. Next,
we propose a decoding algorithm and analyse its complexity, showing that
it operates in time essentially proportional to the cost of linear
algebra in the code length, while correcting up to the unique decoding
radius under the poly-skew metric. Concrete examples are provided
throughout to illustrate the practical applicability of these
constructions.

\section{Skew polynomial ring}
In this section, we recall the definition of the skew polynomial ring and some properties that will be useful throughout the rest of the paper.\\

Consider a finite field $\mathbb{F}_q$ and $\theta$ an automorphism of $\mathbb{F}_q$. The skew polynomial ring (Ore ring \cite{ore1933} ) is the set
\[
 \mathbb{F}_q[X; \theta] =\left\lbrace \sum_{i=0}^{m}a_iX^i \mid a_i \in \mathbb{F}_q \right\rbrace
\]
endowed to the usual addition of polynomials and the multiplication defined by the rule: for $a \in  \mathbb{F}_q$,
\[
X \cdot a = \theta(a) X. \tag{1} \label{prod1}
\]

From formula \ref{prod1}, we deduce a formula for the product of monomials $(a \cdot X^i) \cdot (b\cdot X^j) = a \cdot \theta^{(j)}(b) \cdot X^{i+j}$, where $\theta^{(j)}$ denote $\theta$ compose with itself $j$ times, and then a formula for product distributively. 
If $\theta$ is not the identity, there exist $a \in \Fq$ such that $\theta(a) \neq a$ and then $\Fqx$ is not commutative. Here below, is the formula for the product:
\[ (\displaystyle \sum_{i=0}^{d} a_i \cdot X^i) \cdot (\displaystyle \sum_{j=0}^f b_j \cdot X^j) = \displaystyle \sum_{k=0}^{d+f} (\displaystyle \sum_{i+j=k} a_i  \cdot \theta^{(i)}(b_j)) \cdot X^{k} \label{prod2}\]

Let $A(X)=\displaystyle \sum_{i=0}^a a_i \cdot X^i$ and $B(X)=\displaystyle \sum_{j=0}^b a_j \cdot X^j$ with $0\leq deg(B)=b < deg(A)=a$, it is possible to divide $A$ on the right by $B$ using an euclidian algorithm, computing $Q$ and $R$ with $deg(R)<b$ such that $A(X)=Q(X) \cdot B(X) + R(X)$. In order to present the Euclidian division for skew polynomials, we define few notations. If $P(X) = p_d \cdot X^d + \cdots + p_0$ has degree $d$ (i.e. $p_d \neq 0$) then we denote $LC(P(X))=p_d$ and $LM(P(X))=p_d \cdot X^d$.

\begin{algorithm}
\caption{QUOREM : Euclide algorithm for division on the right} \label{alg1}
\begin{algorithmic}[1]
\Require $A(X)$ and $B(X)$ such that $deg(A(x)) \geq deg(B(x)))$
\State $R(X) \gets A(X)$
\While {$deg(R(X)) \geq deg(B(X))$}
  \State $l=deg(R(X))-deg(B(X))$
  \State $Q(X) \gets Q(X) + \frac{LC(R(X))}{\theta^{l}(LC(B(X))} X^l$ 
  \State $R(X) \gets R(X) - \frac{LC(R(X))}{\theta^{l}(LC(B(X))} X^l \cdot B(X)$
\EndWhile
\State \Return $[Q(X),R(X)]$
\end{algorithmic}
\end{algorithm}

 The output $[Q(X),R(X)]$ of the algorithm \ref{alg1} is the pair satisfying $A(X)=Q(X)\cdot B(X) +R(X)$ with $deg(R(X)) < deg(B(X))$. This implies that $A(X)$ is in the left ideal generate by $B(X)$ if and only if $R(X)=0$. 
 Once we have a Euclide algorithm, it implies that $\Fqx$ is right Euclidian, i.e. evey left ideal is principal. Deducing for the Euclide algorithm an extended Euclidian algorithm, we have the notion of the greatest common right divisor (gcrd) which is the highest degree skew polynomial dividing two polynomials and least common left multiple (lclm) which is the lowest degree polynomial being the left multiple of two polynomials. 
 
 Now, we give the right Extended Euclidian algorithm:
 
 \begin{algorithm}
\caption{REEA : Right extended Euclidian algorithm} \label{alg2}
\begin{algorithmic}[1]
\Require $A(X)$ and $B(X)$ such that $deg(A(x)) \geq deg(B(x)))$ and $T$
\State $R_0(X) \gets A(X)$ and $R_1 \gets B(X)$
\While {$R_1(X) \neq 0$}
  \State $[Q(X),R(X)] \gets {\rm QUOREM}(R_0(X),R_1(X))$
  \State $U \gets U_1$, $U_1 \gets U_0 - Q \cdot U_1$, $U_0 \gets U$
  \State $V \gets V_1$, $V_1 \gets V_0 - Q \cdot V_1$, $V_0 \gets V$
  \State $R_0 \gets R_1$, $R_1 \gets R$
\EndWhile
\If{T=0}
  \State \Return $[U_0,V_0,R_0]$
\ElsIf{T=1}
  \State \Return $U_1 \cdot A(X)$
\EndIf
\end{algorithmic}
\end{algorithm}

If $T=0$ the following returns the coefficient of the Bézout relation of the gcrd $U_0 \cdot A(X) + V_0 \cdot B(X) = R_0(X)$ and if $T=1$ we return the lclm $U_1 \cdot A(X)$. The case $T=0$ is often used to compute modular inverse.

\section{Skew Polynomial Chinese Remainder Theorem codes}

We first do a quick introduction of the needed background on linear codes and then the Chinese remainder theorem in order to define the polynomial skew Chineses remainder codes. 

\subsection{Basics on linear codes}

A linear code is, from a purely formal point of view, a finite-dimensional
subspace of a vector space over a finite field. We recall here the basic
definitions and properties that will be used throughout this work.

\begin{definition}
Let $\Fq$ be a finite field and let $\mathcal{C}$ be a subspace of $\Fq^n$
of dimension $k \leq n$. The space $\mathcal{C}$ is called a \emph{linear
code} of dimension $k$ and length $n$. A matrix $G \in M_{n \times k}(\Fq)$,
viewed as a linear map $\Fq^k \to \Fq^n$, is called a \emph{generator
matrix} of $\mathcal{C}$ if $\operatorname{Im}(G) = \mathcal{C}$.
\end{definition}

In order to measure how far apart two codewords are, and hence the
error-correcting capability of a code, we introduce the notion of weight.

\begin{definition}
A \emph{weight} on $\Fq^n$ is a map $\omega : \Fq^n \rightarrow \mathbb{N}$
satisfying the following properties:
\begin{enumerate}
    \item $\omega(x) = 0$ if and only if $x = 0$;
    \item $\omega(-x) = \omega(x)$ for all $x \in \Fq^n$;
    \item $\omega(x+y) \leq \omega(x) + \omega(y)$ for all $x, y \in \Fq^n$
    (triangle inequality).
\end{enumerate}
\end{definition}

\begin{proposition}
If $\omega : \Fq^n \rightarrow \mathbb{N}$ is a weight, then the map
\[
    d : \Fq^n \times \Fq^n \longrightarrow \mathbb{N}, \qquad
    (x,y) \longmapsto \omega(x-y)
\]
is a distance on $\Fq^n$.
\end{proposition}

\begin{definition}
Let $\mathcal{C}$ be a linear code in $\Fq^n$. The \emph{minimum distance}
of $\mathcal{C}$, denoted $d_{\min}(\mathcal{C})$, is defined as
\[
    d_{\min}(\mathcal{C}) = \min\left\{ d(x,y) \mid x, y \in \mathcal{C},\
    x \neq y \right\}.
\]
\end{definition}

By linearity of $\mathcal{C}$, the minimum distance coincides with the
minimum weight of a nonzero codeword.

\begin{proposition}
\label{prop:min-weight}
We have
\[
    d_{\min}(\mathcal{C}) = \min\left\{ \omega(x) \mid x \in \mathcal{C}
    \setminus \{0\} \right\}.
\]
\end{proposition}

The minimum distance governs the error-correcting capability of a code:
as long as the weight of an error vector remains below half the minimum
distance, unique decoding is guaranteed.

\begin{proposition}
\label{prop:unique-decoding}
Let $c \in \mathcal{C} \subset \Fq^n$ and let $e \in \Fq^n$ be such that
\[
    \omega(e) < \left\lfloor \frac{d_{\min}(\mathcal{C})}{2} \right\rfloor.
\]
Then
\[
    \left\{ y \in \mathcal{C} \;\middle|\; d(y, c+e) \leq
    \left\lfloor \frac{d_{\min}(\mathcal{C})}{2} \right\rfloor \right\}
    = \{ c \}.
\]
\end{proposition}

\subsection{Classical integer and polynomial CRT codes}

Both constructions presented in this section instantiate the same
principle: encoding is performed via the Chinese Remainder Theorem (CRT)
map, which sends a ``small'' element of a ring $A$ to the tuple of its
residues modulo a family of pairwise coprime elements. The redundancy
introduced by taking more moduli than strictly needed for reconstruction
gives the code its error-correcting capability. We present these two
classical cases as a motivation and a reference point for the
generalisation developed in Section~\ref{sec:noncommutative}.

\subsubsection*{Integer CRT codes}

Let $p_1 < p_2 < \cdots < p_n$ be pairwise coprime positive integers and
let $1 \leq k < n$ be an integer. Define
\[
    K \;=\; \prod_{i=1}^{k} p_i
    \qquad \text{and} \qquad
    N \;=\; \prod_{i=1}^{n} p_i.
\]
The \emph{integer CRT encoding map} is
\[
    \Phi \;:\; \mathbb{Z}/K\mathbb{Z} \;\longrightarrow\;
    \mathbb{Z}/p_1\mathbb{Z} \;\times\; \cdots \;\times\;
    \mathbb{Z}/p_n\mathbb{Z}, \qquad
    m \;\longmapsto\;
    \bigl(m \bmod p_1,\; \ldots,\; m \bmod p_n\bigr).
\]
The \emph{integer CRT code} associated to $(p_1, \ldots, p_n)$ and $k$ is
the image
\[
    \mathcal{C}_{\mathrm{int}} \;=\; \Phi\!\left(\mathbb{Z}/K\mathbb{Z}\right)
    \;\subset\;
    \mathbb{Z}/p_1\mathbb{Z} \;\times\; \cdots \;\times\;
    \mathbb{Z}/p_n\mathbb{Z}.
\]
By the Chinese Remainder Theorem, $\Phi$ is injective on $\mathbb{Z}/K\mathbb{Z}$
(since $K \mid N$), so $\mathcal{C}_{\mathrm{int}}$ has cardinality $K$ and
the code has \emph{rate} $\log K / \log N$.

The natural metric on the ambient space is the \emph{polyalphabetic
Hamming weight}: for $\mathbf{c} = (c_1, \ldots, c_n) \in
\prod_{i=1}^n \mathbb{Z}/p_i\mathbb{Z}$,
\[
    \omega_H(\mathbf{c}) \;=\; \#\{\, i \mid c_i \neq 0 \,\}.
\]
The qualifier \emph{polyalphabetic} reflects the fact that different
coordinates belong to different alphabets $\mathbb{Z}/p_i\mathbb{Z}$.
Note that the minimum distance of $\mathcal{C}_{\mathrm{int}}$ with
respect to $\omega_H$ satisfies
$d_{\min} \geq n - k + 1$ when the $p_i$ are distinct primes, in direct
analogy with Reed--Solomon codes.

Integer CRT decoding were introduced by Mandelbaum~\cite{Mandelbaum1976}. An efficient decoding
algorithm, exploiting the extended Euclidean algorithm over $\mathbb{Z}$,
was proposed by Goldreich, Ron and Sudan~\cite{Goldreich2000} and an efficient syndrom decoder was proposed by Li \cite{Li12}. 

\subsubsection*{Polynomial CRT codes}

The polynomial setting is a natural analogue that restores full linearity
over $\Fq$. Let $P_1, \ldots, P_n \in \Fq[X]$ be pairwise coprime
polynomials with $\deg P_1 \leq \cdots \leq \deg P_n$, and let $1 \leq k
< n$. Define
\[
    K \;=\; \prod_{i=1}^{k} P_i
    \qquad \text{and} \qquad
    N \;=\; \prod_{i=1}^{n} P_i,
\]
and write $\kappa = \deg K = \sum_{i=1}^k \deg P_i$ and
$\nu = \deg N = \sum_{i=1}^n \deg P_i$. Consider the $\Fq$-vector space
\[
    \Fq[X]_{<\kappa} \;=\;
    \bigl\{\, F \in \Fq[X] \;\big|\; \deg F < \kappa \,\bigr\},
\]
which has dimension $\kappa$ over $\Fq$. The \emph{polynomial CRT
encoding map} is
\[
    \Psi \;:\; \Fq[X]_{<\kappa} \;\longrightarrow\;
    \Fq[X]/(P_1) \;\times\; \cdots \;\times\; \Fq[X]/(P_n), \qquad
    F \;\longmapsto\;
    \bigl(F \bmod P_1,\; \ldots,\; F \bmod P_n\bigr).
\]
The \emph{polynomial CRT code} associated to $\mathcal{P} = (P_1, \ldots,
P_n)$ and $k$ is
\[
    \mathrm{CRT}(\mathcal{P}, k) \;=\; \Psi\!\left(\Fq[X]_{<\kappa}\right)
    \;\subset\;
    \Fq[X]/(P_1) \;\times\; \cdots \;\times\; \Fq[X]/(P_n).
\]
Since $\Psi$ is $\Fq$-linear and injective (by CRT, as $K \mid N$ in
$\Fq[X]$), $\mathrm{CRT}(\mathcal{P}, k)$ is a linear code of length
$\nu = \sum_{i=1}^n \deg P_i$ and dimension $\kappa = \sum_{i=1}^k \deg
P_i$ over $\Fq$.

The natural weight here is the \emph{weighted Hamming weight}: for
$\mathbf{F} = (F_1, \ldots, F_n) \in \prod_{i=1}^n \Fq[X]/(P_i)$,
\[
    \omega_{wH}(\mathbf{F}) \;=\;
    \sum_{\substack{i=1 \\ F_i \neq 0}}^{n} \deg P_i,
\]
which weights each non-zero coordinate $F_i$ by the ``size'' $\deg P_i$
of its alphabet. This reduces to the standard Hamming weight when all
$\deg P_i$ are equal.

Polynomial CRT codes generalise Reed--Solomon codes. Efficient decoding algorithms based on the extended Euclidean algorithm in
$\Fq[X]$ were studied by Yu and Loeliger~\cite{YuLoeliger12}.

\subsection{Chinese remainder theorem for skew-polynomials and Skew-CRT codes} \label{sec:noncommutative}

We define CRT codes using skew polynomials, which requires applying the Chinese Remainder Theorem in the skew-polynomial setting. 

For any monic polynomial $P \in \mathbb{F}_q[X; \theta]$, We will use the notation $R_P=\mathbb{F}_q[X; \theta]/\mathbb{F}_q[X; \theta](P)$ to denote the $\mathbb{F}_q[X; \theta]$-module of skew polynomials of degree less than $\deg P$ with multiplication right modulo $P$.  

In the commutative ring $\mathbb{F}_q[X]$, to state the Chinese Remainder Theorem, we choose polynomials \(P_1, \ldots, P_n\) that are pairwise coprime. However, this pairwise coprimality condition is not sufficient in the ring $\mathbb{F}_q[X; \theta]$ of skew polynomials, since a skew polynomial may admit several factorizations into irreducible factors. Therefore, we need the notion of \(P\)-independence for polynomials.

\begin{definition}
Let $\Gamma =(P_1,\ldots,P_n)\in {\mathbb{F}_q[X; \theta]}^n$, and $P={\rm lclm}{(P_1,\ldots,P_n)}$. We say that $\Gamma$ is \textbf{$P$-independent}, if  $\deg\left(P\right)=\displaystyle\sum_{i=1}^{n}\deg\left(P_i\right)$.
\end{definition}

The following theorem is the Chinese Remainder Theorem adapted to skew polynomials
\begin{theorem}[Chinese Remainder Theorem] \label{CRT} Let $\Gamma =(P_1,\ldots,P_n)\in {\mathbb{F}_q[X; \theta]}^n$ {$P$-independent},
and $P={\rm lclm}{(P_1,\ldots,P_n)}$.
The map $\psi_\Gamma$ defined by:
\[
\psi_\Gamma =\colon
\left\{
\begin{array}{rcl}
R_{P} & \to & R_{P_1}\times\cdots\times R_{P_n} \\
f & \mapsto & (f \operatorname{mod}_{r}{P_1}, \ldots, f \operatorname{mod}_{r}{P_n})
\end{array}
\right.
\]
is an isomorphism of $\mathbb{F}_q[X; \theta]$-modules. The inverse map is
\begin{equation} \label{fcrt}
{\psi_\Gamma}^{-1} =\colon
\left\{
\begin{array}{rcl}
 R_{P_1}\times\cdots\times R_{P_n} & \to & R_{P} \\
(f_1, \ldots, f_n) & \mapsto & \displaystyle\sum_{i=1}^{n} f_i U_i \operatorname{mod}_{r}{P} 
\end{array}
\right.
\end{equation}
where $U_i= V_iB_i$ with $B_i= {\rm lclm}_{1\leq j\neq i\leq n}(P_j)$ and $V_i \in R_{P_i}$ such that $V_iB_i \equiv 1\operatorname{mod}_{r}{P_i}$
\end{theorem}
 We denote $\mathbb{F}_q[X; \theta]_{<K}= \left\{ f \in \mathbb{F}_q[X; \theta] \mid \deg f< K\right\}$ for $K$ an integer.

\begin{definition}[{skew CRT code}]  Let $\Gamma =(P_1,\ldots,P_n)\in {\mathbb{F}_q[X; \theta]}^n$ {$P$-independent} and $P={\rm lclm}{(P_1,\ldots,P_n)}$. Let $K$ be a positive integer such that $0< K\leq \deg P$.
The \textbf {skew CRT code} of support $\Gamma=(P_1,\ldots,P_n)$ and dimension $K$ is the set
\[
\mathcal{SCRT}_{\Gamma, K} = \psi_{\Gamma}(\mathbb{F}_q[X; \theta]_{<K}).
\]
\end{definition}

\begin{remark}\label{Rem1} The skew CRT codes generalize some well known families of codes.
\begin{itemize}
    \item When $\deg P_1=\cdots=\deg P_n=1$, then the skew CRT code is a skew Reed-Solomon code \cite{boucher2014}.
    \item When $\theta=Id$, then the skew CRT code is a polynomial CRT code.
    \item When $\deg P_1=\cdots=\deg P_n=1$ and $\theta=Id$, then the skew CRT code is a Reed-Solomon code.
\end{itemize}
\end{remark}
\medskip

\section{Poly-Skew metric}

In this section, we introduce a metric for which the skew CRT codes are optimal and which will be used later for decoding. Before doing that, we recall the definition of some well-known metrics.

For $y = (y_1, \ldots, y_n) \in \mathbb{F}_q^n$, the Hamming weight of $y$ noted $w_H(y)$ is the number
of non-zero coordinates of $y$:
\[
w_H(y) := \#\{ i \in \{1, \ldots, n\} \mid y_i \neq 0 \}.
\]

The skew Reed--Solomon code is a particular class of skew CRT codes, as noted in Remark~\ref{Rem1}. The metric that is usually considered optimal for skew Reed--Solomon codes is the skew metric first introduced in \cite{martinez2018}, defined as follows.

Let $\alpha = (\alpha_1, \ldots, \alpha_n) \in \mathbb{F}_q^n$. One say that $\alpha$ is $P$-independent if $(X-\alpha_1, \ldots, X-\alpha_n)$ is $P$-independent. For all $y = (y_1, \ldots, y_n) \in \mathbb{F}_q^n$, the skew weight of $y$ associated to $\alpha$ is defined as
$$w_\alpha(y)=\deg \left(  {\rm lclm}_{y_i\neq 0 }\left(X-\alpha_i\frac{\theta(y_i)}{y_i}\right)\right).$$

Another interpretation of the skew weight was given in \cite[Lemma~1]{boucher2020}. We have

\begin{equation}\label{skewegal}
w_\alpha(y) = \deg(P) - \deg\left( \mathrm{gcrd}\left(P, \psi_\Gamma^{-1}(y)\right) \right),
\end{equation}
where $P = \mathrm{lclm}(X - \alpha_1, \ldots, X - \alpha_n)$ and $\Gamma=(X - \alpha_1, \ldots, X - \alpha_n)$.\\
Throughout the remainder of the paper, given $\Gamma = (P_1, \ldots, P_n)\in \mathbb{F}_q[X; \theta]^n $ $P$-independent, we denote 
$$\mathcal{L} = R_{P_1} \times \cdots \times R_{P_n}$$
where $ R_{P_i}={\mathbb{F}_q[X; \theta]}/{\mathbb{F}_q[X; \theta]}P_i $. \\Following the idea behind equation~\ref{skewegal}, we extend the definition of the skew weight to elements of $\mathcal{L}$ as follows.

\begin{definition}\label{sweight} Consider $\Gamma =(P_1,\ldots,P_n)\in {\mathbb{F}_q[X; \theta]}^n$ {$P$-independent},
and $P={\rm lclm}{(P_1,\ldots,P_n)}$.
Let $y=(y_1,\ldots,y_n) \in \mathcal{L}$. The \textbf{ skew poly-weight} of $y$ associated to $\Gamma$ is
$$w_{\Gamma}(y)=\deg(P)-\deg\left( {\rm gcrd}\left(P,\psi_\Gamma^{-1}(y)\right)\right).$$

\end{definition}

To prove that the skew poly-weight induces a metric on $\mathcal{L}$, we need the following technical lemma.

\begin{lemma}\label{Lsum} Let $P \in \mathbb{F}_q[X; \theta]$ such that $\deg P \geq 0$.
For all $f$ and $g$ in $R_{P}$ we have $$\deg\left( {\rm gcrd}\left(P,f\right)\right) + \deg\left( {\rm gcrd}\left(P,g\right)\right)\leq \deg\left(P\right) + \deg\left( {\rm gcrd}\left(P,f+g\right)\right).$$
\end{lemma}

\begin{proof}
\vspace{0.5cm}
Since ${\rm lclm}\left( {\rm gcrd}\left(P,f\right), {\rm gcrd}\left(P,g\right)\right)$ is a right divisor of $P$, one has :
$\deg {\rm lclm}\left( {\rm gcrd}\left(P,f\right), {\rm gcrd}\left(P,g\right)\right) \leq \deg P$. Therefore,
$$\deg\left( {\rm gcrd}\left(P,f\right)\right) + \deg\left( {\rm gcrd}\left(P,g\right)\right)\leq \deg\left(P\right) + \deg\left( {\rm gcrd}\left({\rm gcrd}\left(P,g\right),{\rm gcrd}\left(P,f\right)\right)\right).$$
As ${\rm gcrd}\left({\rm gcrd}\left(P,g\right),{\rm gcrd}\left(P,f\right)\right)$ divides ${\rm gcrd}\left(P,f+g\right)$, we have the result. 
\end{proof}

\begin{proposition} Let $\Gamma = (P_1, \ldots, P_n)\in \mathbb{F}_q[X; \theta]^n $ $P$-independent and $P={\rm lclm}{(P_1,\ldots,P_n)}$.
The map \[
d_{\Gamma} =\colon
\left\{
\begin{array}{rcl}
\mathcal{L}\times\mathcal{L} & \to &\mathbb{N} \\
(u,  v) & \mapsto & w_{\Gamma}(u-v)
\end{array}
\right.
\] is a distance.
\end{proposition}

\begin{proof} Given $u,v, w \in \mathcal{L}$, we will denote $U=\psi_\Gamma^{-1}(u)$, $V=\psi_\Gamma^{-1}(v)$ and $W=\psi_\Gamma^{-1}(w)$. We show now that $d_\Gamma$ is a distance:

\begin{itemize}
    \item  $\forall u,v \in \mathcal{L}, d_{\Gamma}(u,v)\geq 0$
    
    We have $d_{\Gamma}(u,v)=\deg(P) - \deg\left( \mathrm{gcrd}\left(P, \psi_\Gamma^{-1}(u-v)\right) \right)$. Since $\mathrm{gcrd}\left(P, \psi_\Gamma^{-1}(u-v)\right) $ divides $P$ on right, we have $\deg(P) \geq \deg\left( \mathrm{gcrd}\left(P, \psi_\Gamma^{-1}(u-v)\right) \right)$. Thus $d_{\Gamma}(u,v)\geq0$

     \item  $\forall u,v \in \mathcal{L}, d_{\Gamma}(u,v)= 0 \Leftrightarrow u=v$
     
     Assume that $d_{\Gamma}(u,v)= 0$. That imply $P$ divides $\psi_\Gamma^{-1}(u-v)=U-V$ on right. Since $\deg(U-V) < \deg P$, then $U-V=0$. This is equivalent to $u=v$.
\item  $\forall u,v \in \mathcal{L}, d_{\Gamma}(u,v)= d_{\Gamma}(v,u)$

As $\psi_\Gamma^{-1}(u-v)=-\psi_\Gamma^{-1}(v-u)$, then $\mathrm{gcrd}\left(P, \psi_\Gamma^{-1}(u-v)\right) =\mathrm{gcrd}\left(P, \psi_\Gamma^{-1}(v-u)\right)$. Therefore $d_{\Gamma}(u,v)= d_{\Gamma}(v,u)$.

\item  $\forall u,v \in \mathcal{L}, d_{\Gamma}(u,v)\leq d_{\Gamma}(u,\omega)+d_{\Gamma}(\omega,v)$

According to Lemma~\ref{Lsum}, we have 

$
\deg\left( \mathrm{gcrd}\left(P,U-W\right)\right)
+
\deg\left( \mathrm{gcrd}\left(P,W-V\right)\right)
\leq
\deg(P)
+
\deg\left( \mathrm{gcrd}\left(P,U-V\right)\right)
$ because $(U-V)=(U-W) +(W-V)$. Thus 

$
-\deg\left( \mathrm{gcrd}\left(P,U-V\right)\right)
\leq
\deg(P)-\deg\left( \mathrm{gcrd}\left(P,U-W\right)\right)
-\deg\left( \mathrm{gcrd}\left(P,W-V\right)\right)
$
By adding $\deg P$ to both sides of this inequality, we obtain
$d_{\Gamma}(u,v)\leq d_{\Gamma}(u,\omega)+d_{\Gamma}(\omega,v)$. 

\end{itemize}
\end{proof}
In the following example, we construct a word $y$ and determine its poly-skew weight for a given $\Gamma$.
\begin{ex}
Consider $\mathbb{F}_{3^{10}}=\mathbb{F}_{3}(a)$ where $a^{10} + 2a^6 + 2a^5 + 2a^4 + a + 2=0$ and $\mathbb{F}_{3^{10}}[X,\theta]$ with $\theta\colon b\mapsto b^3$ the Frobenius automorphism.
 Let $\Gamma=(P_1,P_2,P_3,P_4,P_5)\in \mathbb{F}_{3^{10}}[X,\theta]^5$ such that\\
  
  \noindent $P_1=X^2 + (a^8 + 2a^6 + 2a^5 + 2a^4 + a^3 + 2a + 1)X + 2a^8 + a^7 + a^6 + a^5 + a^4 + a^3 + a + 2,$
  \vspace{0.35 cm}
  
 \noindent $P_2=X^2 + (2a^8 + a^7 + 2a^6 + a^4 + 2a^3 + 2a^2 + a + 1)X + 2a^9 + 2a^8 + 2a^6 + 2a^5 + 2a^4 + a^2 + 1,$
 
 \vspace{0.35 cm}
 \noindent  $P_3=X^2 + (a^8 + 2a^5 + a^4 + 2a^2 + a + 1)X + a^8 + 2a^7 + a^4 + 2a^2 + 2,$
  
  \vspace{0.35 cm}
 \noindent $P_4=X^2 + (a^9 + 2a^8 + a^6 + a^3 + a^2 + 2a)X + a^8 + 2a^7 + 2a^5 + a^2 + 2a + 1,$

 \vspace{0.35 cm}
 \noindent $P_5=X^2 + (a^9 + 2a^7 + 2a^4 + 2a^3 + a + 1)X + 2a^9 + 2a^8 + a + 2$
 \vspace{0.35 cm}\\
 
 \noindent We get
 $P={\rm lclm}(P_1,P_2,P_3,P_4,P_5)=X^{10}+2$.
 
  \vspace{0.35 cm}
  \noindent Let $y=( (a^8 + a^7 + 2a^6 + a^5 + 2a^4 + a + 1)X + a^8 + 2a^6 + a^5 + a^4 + a^2 + 2a + 2, (a^9 + 2a^8 + a^7 + 2a^5 + a^4 + a^3 + a^2 + a)X + 2a^8 + a^5 + 2a^4 + 2a^3 + 2a^2, (a^8 + 2a^7 + a^6 + 2a^4 + a^3 + a^2 + 2)X + 2a^8 + a^5 + 2a^2 + 1  ,(a^8 + a^7 + a^6 + 2a^5 + a^4 + a^3 + a^2 + 2a + 1)X + 2a^9 + a^8 + 2a^7 + a^6 + 2a^5 + a^3 + a^2 + 2a + 2,  (a^9 + 2a^7 + a^6 + 2a^4 + a^3 + a)X + 2a^8 + a^7 + a^6 + a^4 + 2a^3 + a^2)$.

 \vspace{0.35 cm}
  \noindent We have $\psi _{\Gamma}^{-1}(y)=X^9 + (a^9 + a^7 + a^3 + 2a^2 + 2a)X^8 + (a^9 + a^8 + 2a^4)X^7 + (a^9 + a^6 + a^5 + 2a^3 + 2a^2 + a)X^6 + (2a^9 + 2a^8 + a^7 + 2a^3 + a^2)X^5 + (2a^9 + a^4 + 2a^2 + 2a + 1)X^4 + (a^9 + 2a^8 + 2a^7 + a^5 + 2a^3 + 2a + 1)X^3 + (2a^9 + a^8 + 2a^7 + a^6 + a^5 + a^4 + 2a^2 + 2)X^2 + (2a^7 + 2a^6 + 2a^5 + 2a^4 + 2a^3 + 2a^2 + 2)X + a^9 + a^8 + 2a^7 + a^5 + 2a^4 + a^2 + 1$
 
\vspace{0.35 cm}
 \noindent and ${\rm gcrd}\left(P,\psi _{\Gamma}^{-1}(y) \right)=X^8 + (a^8 + a^6 + 2a^4 + 2a^3 + 2)X^7 + (2a^7 + a^4 + 2a + 1)X^6 + (2a^9 + 2a^8 + 2a^7 + a^6 + a^5 + 2a^3 + 2a + 1)X^5 + (2a^9 + a^8 + a^7 + 2a^5 + a^4 + 2a^3 + 2a)X^4 + (a^9 + a^8 + a^7 + a^4 + 2a^3 + 1)X^3 + (2a^9 + 2a^7 + 2a^6 + a^5 + a^4 + a^3 + a^2 + 2a + 1)X^2 + (2a^9 + 2a^8 + a^7 + 2a^6 + a^4 + a^3 + 2a^2 + 2a)X + a^9 + 2a^8 + 2a^7 + 2a^6 + a^5 + a^4 + a^3 + a^2 + a + 1$.

\vspace{0.35 cm}
 \noindent Therefore $w_{\Gamma}(y) =10-8=2$.

\end{ex}
We now turn our attention to the minimum distance of a skew CRT code with respect to the poly-skew metric in order to determine its error-correcting capability
\begin{proposition}

Let $\Gamma = (P_1, \ldots, P_n) \in \mathbb{F}_q[X; \theta]^n$ $P$-independent and $P = \mathrm{lclm}(P_1, \ldots, P_n)$. Let
$N=\deg(P)$.
Consider the skew CRT codes $ \mathcal{SCRT}_{\Gamma, K}$.

\begin{enumerate}
    \item The minimum distance $d_{\rm min}$ of $ \mathcal{SCRT}_{\Gamma, K}$ verify the following inequality  $$d_{\rm min}\geq N-K+1.$$
\item If $g \in \mathbb{F}_q[X; \theta] $ is a right divisor of $P$ with $\deg g= K-1$, then
$$d_{\rm min}= N-K+1.$$
\end{enumerate}
\end{proposition}

\begin{proof} Recall that, $$d_{\min}=\min \lbrace w_{\Gamma}(c)\mid c \in \mathcal{SCRT}_{\Gamma, K}, c\neq 0 \rbrace.$$
\begin{enumerate}
    \item For all $c \in \mathcal{SCRT}_{\Gamma, K}$, there exist $f\in \mathbb{F}_q[X; \theta]_{<K}$ such that $f = \psi_\Gamma^{-1}(c)$. We have $w_{\Gamma}(c)=\deg(P)-\deg\left( {\rm gcrd}\left(P,\psi_\Gamma^{-1}(c)\right)\right)=\deg(P)-\deg\left( {\rm gcrd}\left(P,f\right)\right).$ Since $\deg f\leq K-1< \deg P$, we have $\deg\left( {\rm gcrd}\left(P,f\right)\right)\leq K-1$. Thus, for all $c \in \mathcal{SCRT}_{\Gamma, K}$, $w_{\Gamma}(c)\geq deg P-(K-1)$.
    
\item Assume that there is $g \in \mathbb{F}_q[X; \theta] $, a right divisor of $P$ such that $\deg g= K-1$. We have $\psi_\Gamma(g) \in \mathcal{SCRT}_{\Gamma, K}$ and  $w_{\Gamma}(\psi_\Gamma(g))=\deg(P)-\deg\left( {\rm gcrd}\left(P,g\right)\right)= \deg(P)-\deg\left(g\right)$. Therefore $w_{\Gamma}(\psi_\Gamma(g))=N-K+1$

\end{enumerate}
\end{proof}











\section{Error model}
\label{sec:error_model}

A central ingredient of any decoding algorithm is a concrete description
of the errors that can occur during transmission, i.e.\ an
\emph{error model}. In the context of the poly-skew metric, an error
vector is an element $e \in \mathcal{L} = R_{P_1} \times \cdots \times
R_{P_n}$ of poly-skew weight at most $t$. The error model therefore
reduces to the question of how to efficiently construct such vectors,
which is what this section addresses.

\paragraph{Special case: skew Reed--Solomon setting.}
When $\Gamma = (X - a_1, \ldots, X - a_n)$ is $P$-independent, the
poly-skew metric coincides with the skew metric, and the construction of
a word of prescribed skew weight $t$ is well understood and
efficient~\cite{nouetowa2025}. This special case thus provides both a
sanity check and a baseline for the general construction below.

\paragraph{General case.}
When $\Gamma = (P_1, \ldots, P_n)$ is $P$-independent with moduli of
arbitrary degree, we use the following approach. The key observation is
that if $M$ is a right divisor of $P = \operatorname{lclm}(P_1, \ldots,
P_n)$ of degree $N - t$, and $u$ is a skew polynomial of degree $t-1$,
then the product $uM$ has degree at most $N - 1$ and its image under the
CRT map has poly-skew weight at most $t$. The poly-skew weight equals
exactly $t$ when $\operatorname{gcrd}(P, uM) = M$, i.e.\ when $u$
introduces no additional right common factor with $P$ beyond $M$, which
holds for a generic choice of $u$. This leads to the following
algorithm.

\begin{algorithm}[H]
\caption{Construction of a word of poly-skew weight $t$}
\label{alg:error_generation}
\textbf{Require:} $\Gamma = (P_1, \ldots, P_n)$ $P$-independent,
$P = \operatorname{lclm}(P_1, \ldots, P_n)$, and an integer $t \geq 1$.\\
\textbf{Ensure:} $y \in \mathcal{L} = R_{P_1} \times \cdots \times
R_{P_n}$ with $\operatorname{wt}_{\mathrm{ps}}(y) \leq t$.
\begin{algorithmic}[1]
\State Construct a right divisor $M$ of $P$ of degree $N - t$.
\State Choose $u \in \Fq[X;\theta]$ of degree $t - 1$.
\State Compute
\[
    y = \bigl((uM) \bmod P_1,\; \ldots,\; (uM) \bmod P_n\bigr).
\]
\State \Return $y$.
\end{algorithmic}
\end{algorithm}

The word $y$ thus constructed satisfies
$\operatorname{wt}_{\mathrm{ps}}(y) \leq t$. Moreover,
$\operatorname{wt}_{\mathrm{ps}}(y) = t$ if and only if
$\operatorname{gcrd}(P, uM) = M$, a condition that holds for a generic
choice of $u$ and can be verified efficiently. This algorithm therefore
provides a practical and explicit error model for the poly-skew metric,
which is used in the decoding algorithm of
Section~\ref{sec:decoding}.

\section{Decoding algorithm }
\label{sec:decoding}
In this section, we give an efficient algorithm that enables decoding of Skew CRT codes in polynomial time with respect to the poly-skew metric.

Let $\Gamma = (P_1, \ldots, P_n) \in \mathbb{F}_q[X; \theta]^n$ $P$-independent, $P = \mathrm{lclm}(P_1, \ldots, P_n)$ and $\mathcal{L} = R_{P_1} \times \cdots \times R_{P_n}$. Let
$N=\deg(P)$ and $K$ an integer such that $0<K\leq N$.
Consider the skew CRT codes $ \mathcal{SCRT}_{\Gamma, K}$.
Let $y = c + e$ a receive word, where $c \in \mathcal{SCRT}_{\Gamma, K}$ and $e = (e_1, \ldots, e_n) \in \mathcal{L}$ such that $w_{\Gamma}(e) = t$. Since the minimum distance of $\mathcal{SCRT}_{\Gamma,K}$ satisfies $d_{\min} \geq N-K+1$, we can consider
\[
t \leq \frac{N - K}{2}.
\]
Since $c \in \mathcal{SCRT}_{\Gamma, K}$, there exist $f\in \mathbb{F}_q[X; \theta]_{<K}$ such that $f = \psi_\Gamma^{-1}(c)$. The aim is to recover $f$. Let'us denote 
$E = \psi_\Gamma^{-1}(e) \in R_{P}$ and $Y = \psi_\Gamma^{-1}(y) \in R_{P}$, we get \begin{equation}\label{Eq1}
   Y = f + E.
\end{equation}
Consider $V$ and $\Lambda$ in $\mathbb{F}_q[X; \theta]$ such that
    $ V E = \mathrm{lclm}(E,P)$  and $ \Lambda = \mathrm{gcrd}(E,P)$.
    
\medskip

\begin{proposition} We have
\begin{equation}\label{Eq2}
    VY  \operatorname{mod}_{r}{P}  =  Q,
\end{equation}

with $Q = Vf$ and $\deg(Q) \leq K - 1 + t$.
\end{proposition}

\begin{proof}
As  $ V E = \mathrm{lclm}(E,P)$, we have $
    V E \equiv 0 \operatorname{mod}_{r}{P}.
    $ and $\deg(V) = \deg(P) - \deg(\Lambda)$. According to Definition~\ref{sweight} $w_\Gamma(e)= \deg(P) - \deg(\Lambda)$. Therefore $\deg(V) =w_\Gamma(e)=t$.
Since $VY = VE + Vf$, we obtain
\[
VY \equiv Q \operatorname{mod}_{r}{P},
\]
where $Q = Vf$ and $\deg(Q)=\deg(V) +\deg(f) \leq K - 1 + t$. As $t\leq \frac{N-K}{2}$, then $\deg(Q) \leq K - 1 + \frac{N-K}{2} \leq N -1$. Therefore, we get 
\[
VY \operatorname{mod}_{r}{P} = Q.
\]
\end{proof}

Now, we want to solve the equation (\ref{Eq2}) using the linear algebra, to determine $Q$ and $V$. Assume that $V =\displaystyle \sum_{i=0}^{t} V_i X^i$ and  $Q = \displaystyle\sum_{i=0}^{K-1+t} Q_i X^i$. As $VY =\displaystyle \sum_{i=0}^{t} V_i X^iY$, we get 
$VY \operatorname{mod}_{r}{P} =\displaystyle \sum_{i=0}^{t} V_i \left(X^iY \operatorname{mod}_{r}{P}\right)$. Let us denote 
\[
M_Y =
\left[
Y  \operatorname{mod}_{r}{P} \mid
X \ast Y  \operatorname{mod}_{r}{P} \mid
\cdots \mid
X^t \ast Y  \operatorname{mod}_{r}{P} 
\right],
\]
the $N \times (t+1)$ matrix whose columns are the coefficient vectors of the corresponding skew polynomials.
The relation $VY \operatorname{mod}_{r}{P} = Q$ is equivalent to
\begin{equation}\label{Eq3}
M_Y
\begin{bmatrix}
V_0 \\
\vdots \\
V_t
\end{bmatrix}
=
\begin{bmatrix}
Q_0 \\
\vdots \\
Q_{K-1+t} \\
0 \\
\vdots \\
0
\end{bmatrix}
.\end{equation}

Let $I_{N}$ be the $N \times N$ identity matrix. We denote by $I_{N, K+t}$ the $N\times(K+t)$ matrix obtained by truncating $I_{N}$ to its first $K+t$ columns.
The equation system (\ref{Eq3}) is equivalent to
\begin{equation}\label{Eq4}
\left[M_Y \mid -I_{(N, K+t)}\right]
\begin{bmatrix}
V_0 \\
\vdots \\
V_t \\
Q_0 \\
\vdots \\
Q_{K-1+t}
\end{bmatrix}
=
\begin{bmatrix}
0 \\
\vdots \\
0
\end{bmatrix}.
\end{equation}

Therefore, $V$ and $Q$ can be determined by solving the linear system (\ref{Eq4}). Thus, we can compute $f= Q{/}_{\ell}V$ as the quotient of the left division of $Q$ by $V$. The different steps leading to the computation of $f$ are summarized in the following algorithm.

\begin{algorithm}[H]
\caption{Decoding algorithm of skew CRT codes }
\label{alg:skew_CRT_decode}

\textbf{Require:} 
 $y = c + e$ with
$w_\Gamma(e) \le t $, 
$c = \left(f \operatorname{mod}_{r}{P_1}, \ldots, f \operatorname{mod}_{r}{P_n}\right)$ and $f \in \mathbb{F}_q[X; \theta]$.

\textbf{Ensure:} $f$

\begin{algorithmic}[1]
\State $Y \gets \psi_\Gamma^{-1}(y) \in R_{P}$.
\If{$\deg Y<K$}
  \State \Return $Y$
\Else
  \State $M_Y \gets
\left[
Y \operatorname{mod}_{r}{P} \mid
X \ast Y \operatorname{mod}_{r}{P} \mid
\cdots \mid
X^t \ast Y \operatorname{mod}_{r}{P}
\right]
$
\State $B\gets\left[M_Y \mid -I_{(N,K+t)}\right]$
\State {Computation of} ${\bf u} \neq 0$ such that $B{\bf u}=0$
\State $V \gets\displaystyle \sum_{i=0}^{t} u_i X^i$ and  $Q \gets \displaystyle\sum_{i=0}^{K-1+t} u_{t+1+i} X^i$
\State $f \gets Q{/}_{\ell}V$
\EndIf

\State \Return $f$

\end{algorithmic}
\end{algorithm}

We need the following technical lemma to prove that the Algorithm~\ref{alg:skew_CRT_decode} is correct.

\begin{lemma}\label{wgf}
Consider $\Gamma =(P_1,\ldots,P_n)\in {\mathbb{F}_q[X; \theta]}^n$ {$P$-independent},
and $P={\rm lclm}{(P_1,\ldots,P_n)}$. Let $f,g \in{\mathbb{F}_q[X; \theta]}$ then  $$w_{\Gamma}\left(gf \operatorname{mod}_{r}{P_1}, \ldots, gf \operatorname{mod}_{r}{P_n}\right)\leq w_{\Gamma}\left(f \operatorname{mod}_{r}{P_1}, \ldots, f \operatorname{mod}_{r}{P_n}\right).$$

\end{lemma}

\begin{proof} Let'us denote $V_f=\left(f \operatorname{mod}_{r}{P_1}, \ldots, f \operatorname{mod}_{r}{P_n}\right)$ and\\
$V_{gf}=\left(gf \operatorname{mod}_{r}{P_1}, \ldots, gf \operatorname{mod}_{r}{P_n}\right)$.
According to Definition~\ref{sweight}, we have $w_{\Gamma}\left(V_f \right)=\deg(P)-\deg\left( {\rm gcrd}\left(P,f\right)\right)$ and $w_{\Gamma}\left(V_{gf} \right)=\deg(P)-\deg\left( {\rm gcrd}\left(P,gf\right)\right)$. Therefore $w_{\Gamma}\left(V_{gf} \right)=w_{\Gamma}\left(V_f \right)+\deg\left( {\rm gcrd}\left(P,f\right)\right)-\deg\left( {\rm gcrd}\left(P,gf\right)\right)\leq w_{\Gamma}\left(V_f \right).$
\end{proof}

\begin{proposition} Decoding Algorithm~\ref{alg:skew_CRT_decode} is correct.
\end{proposition}

\begin{proof} Consider $Z(X)=Q-Vf \in {\mathbb{F}_q[X; \theta]}$ and $U \in {\mathbb{F}_q[X; \theta]}$ such that $UZ= {\rm lclm}(P,Z)$. We want to prove that $Z$ is equal to zero. We have, $\deg U= \deg(P)-\deg\left( {\rm gcrd}\left(P,Z\right)\right)=w_{\Gamma}\left(Z \operatorname{mod}_{r}{P_1}, \ldots, Z \operatorname{mod}_{r}{P_n}\right)$. 
Since $VY-Q\equiv 0 \operatorname{mod}_{r}{P}$, we get $VY \operatorname{mod}_{r}{P_i}=Q \operatorname{mod}_{r}{P_i}$ for all $i \in \{1, \ldots, n\}$. Thus,  $Z\operatorname{mod}_{r}{P_i}=Q-Vf\operatorname{mod}_{r}{P_i}=V(Y-f)\operatorname{mod}_{r}{P_i}$ for all $i \in \{1, \ldots, n\}$. According to Lemma~\ref{wgf}, one has $w_{\Gamma}\left(V(Y-f) \operatorname{mod}_{r}{P_1}, \ldots, V(Y-f) \operatorname{mod}_{r}{P_n}\right)\leq t$ because $w_{\Gamma}\left(Y-f \operatorname{mod}_{r}{P_1}, \ldots, Y-f \operatorname{mod}_{r}{P_n}\right)\leq t$. Therefore $\deg U\leq t$. Since $\deg Z\leq K-1+t$, one has $\deg UZ\leq K-1+2t\leq N-1$.
Furthermore, as $UZ\operatorname{mod}_{r}{P}=0$ then $UZ=0$ and $Z=0$.
\end{proof}

To illustrate our algorithm, we present the following example, which was constructed using SageMath.
\begin{ex} Consider $\mathbb{F}_{2^{11}}=\mathbb{F}_{2}(a)$ where $a^{11} + a^2 +1=0$ and $\mathbb{F}_{2^{11}}[X,\theta]$ with $\theta\colon b\mapsto b^2$ the Frobenius automorphism. Let $\Gamma=(P_1,P_2,P_3,P_4,P_5)\in \mathbb{F}_{2^{11}}[X,\theta]^5$ such that\\

  \noindent $P_1=[X^2 + (a^9 + a^6 + a^5 + a^2 + 1)X + a^{10} + a^8 + a^7 + a^5 + a^4 + a + 1,$
  \vspace{0.35 cm}
  
 \noindent $P_2= X^2 + (a^7 + a^6 + a^4 + a^2)X + a^{10} + a^8 + a^7 + a^5 + a^4 + a^3 + a^2 + a,$
 
 \vspace{0.35 cm}
 \noindent  $P_3=X^2 + (a^6 + a^3 + a + 1)X + a^9 + a^7 + a^5 + a^4 + a^2 + 1,$
  
  \vspace{0.35 cm}
 \noindent $P_4= X^2 + (a^9 + a^6 + a^4 + a^3 + a)X + a^9 + a^8 + a^7 + a^6 + a^2 + a + 1,$

 \vspace{0.35 cm}
 \noindent $P_5=X^2 + (a^7 + a^6 + 1)X + a^9 + 1$
 \vspace{0.35 cm}
 
 \noindent 
 $P={\rm lclm}(P_1,P_2,P_3,P_4,P_5)=X^{10} + (a^{10} + a^7 + a^5 + a^3 + a + 1)X^9 + (a^9 + a^5 + 1)X^8 + (a^{10} + a^9 + a^7 + a^6 + a^5 + a^3)X^7 + (a^{10} + a^9 + a^8 + a^4 + a^2 + a)X^6 + (a^7 + a)X^5 + (a^{10} + a^9 + a^8 + a^6 + a^4 + a)X^4 + (a^{10} + a^7 + a^4 + a^3 + a^2 + 1)X^3 + (a^{10} + a^9 + a^8 + a^6 + a^2 + a + 1)X^2 + (a^9 + a^8 + a^7 + a^6 + a + 1)X + a^{10} + a^9 + a^7 + a^6 + a^4 + a^2$.

 \vspace{0.35 cm}
 \noindent $e=((a^9 + a^6 + a^4 + a^3)X + a^8 + a^7 + a^5 + 1,(a^9 + a^7 + a^6 + a^5 + a^3 + a)X + a^{10} + a^6 + a^5 + a^4 + a + 1, (a^{10} + a^8 + a^7 + a^6 + a^5 + a^4)X + a^8 + a^7 + a^4 + a^3 + a^2, (a^6 + a^5 + a^4 + a^3 + 1)X + a^{10} + a^9 + a^6 + a^5 + a^4 + a^2 + a, (a^7 + a^6 + a^4)X + a^8 + a^7 + a^6 + a^5 + a^4 + a^3 + 1 )$
 
 \vspace{0.35 cm}
  \noindent $f=X^3+aX^2+(a^7)X+1$
  
  \vspace{0.35 cm}
  \noindent $c=\psi_\Gamma(f)=((a^9 + a^8 + a^5 + a^4 + a)X + a^10 + a^7 + a^6 + a^5 + a^4 + a^3 + a^2 + a + 1, (a^10 + a^7 + a^6 + a^5 + 1)X + a^5 + a + 1, (a^10 + a^9 + a^7 + a^6 + a^3 + a^2)X + a^7 + a^6 + a^2 + a, (a^10 + a^9 + a^8 + a^7 + a^6 + a^5 + a)X + a^6 + a^5 + a^3, (a^9 + a^7 + a^6 + a^5 + a^3 + a^2 + a + 1)X + a^9 + a^3)$

  \vspace{0.35 cm}
  \noindent Let $y=c+e=((a^8 + a^6 + a^5 + a^3 + a)X + a^{10} + a^8 + a^6 + a^4 + a^3 + a^2 + a, (a^{10} + a^9 + a^3 + a + 1)X + a^{10} + a^6 + a^4, (a^9 + a^8 + a^5 + a^4 + a^3 + a^2)X + a^8 + a^6 + a^4 + a^3 + a, (a^{10} + a^9 + a^8 + a^7 + a^4 + a^3 + a + 1)X + a^{10} + a^9 + a^4 + a^3 + a^2 + a, (a^9 + a^5 + a^4 + a^3 + a^2 + a + 1)X + a^9 + a^8 + a^7 + a^6 + a^5 + a^4 + 1)$.

 \vspace{0.35 cm}
  \noindent We have $Y=\psi_\Gamma^{-1}(y)=X^9 + (a^{10} + a^9 + a^8 + a^5 + a^4 + a^3)X^8 + (a^{10} + a^9 + a^8 + a^7 + a^6 + a^5 + a^4 + a + 1)X^7 + (a^{10} + a^5 + a^2 + a)X^6 + (a^{10} + a^8 + a^7 + a^5 + 1)X^5 + (a^7 + a^5 + a^4 + a^2 + a)X^4 + (a^7 + a^6 + a^5 + a^4 + a^3 + a^2 + a)X^3 + (a^{10} + a^9 + a^6 + a^5 + a^3 + a^2 + a + 1)X^2 + (a^{10} + a^9 + a^8 + a^5 + a^4 + a^2 + a + 1)X + a^7 + a^6 + a^5 + a^4 + a^3 + a^2 + a$

\vspace{0.35 cm}
 \noindent $V=(a^{10} + a^9 + a^5 + a^4 + 1)X^2 + (a^{10} + a^9 + a^7 + a^6 + a^4 + a^2 + a)X + 1$

 \vspace{0.35 cm}
 \noindent $ Q=(a^{10} + a^9 + a^5 + a^4 + 1)X^5 + (a^{10} + a^8 + a^7 + a^6 + a^5 + a^4 + a^3 + a)X^4 + (a^{10} + a^7 + a^4 + a^3 + a^2 + a)X^3 + (a^9 + a^5 + a^3 + a)X^2 + (a^{10} + a^9 + a^6 + a^4 + a^2 + a)X + 1$
 
\vspace{0.35 cm}
 \noindent The quotient in the left division of $Q$ by $V$ is equal to $f$.

\end{ex}


\begin{remark}\label{Rem2} It is important to note that the decoding algorithm for Skew CRT codes in the poly-skew metric, given in this paper as Algorithm~\ref{alg:skew_CRT_decode}, can also be used to decode certain well-known families of codes, such as Gabidulin codes in the rank metric and linearized Reed-Solomon codes in the sum-rank metric. Indeed, skew Reed-Solomon codes form a special case of skew CRT codes. Moreover, decoding Gabidulin codes in the rank metric, as well as decoding linearized Reed-Solomon codes in the sum-rank metric, can be reduced to decoding a particular family of skew Reed-Solomon codes in the skew metric.

Furthermore, when $\theta = \mathrm{identity}$, the algorithm also allows the decoding of Reed-Solomon codes in the Hamming metric and CRT codes in the weighted Hamming metric.

\end{remark}

\section{Complexity of Algorithm~\ref{alg:skew_CRT_decode}}
\label{sec:complexity}

We analyse the computational cost of Algorithm~\ref{alg:skew_CRT_decode},
measured in terms of arithmetic operations in the prime subfield $\Fp$,
where $q = p^m$ and $\Fp$ denotes the prime subfield of $\Fq$. The
analysis proceeds in two steps: we first collect known complexity bounds
for elementary operations on skew polynomials in $\Fq[X;\theta]$, then
derive the overall cost of the decoding algorithm by identifying its
dominant steps.

\subsection*{Arithmetic complexity of skew polynomials}

We follow the notation and results of Caruso and Le Borgne~\cite{CARUSOLEBORGNE}.
We denote by $\SM{d}{m}$ the number of $\Fp$-operations required to
multiply two skew polynomials in $\Fq[X;\theta]$ of degree at most $d$.
Throughout, we assume that every arithmetic operation in $\Fq$ can be
performed in quasi-linear time, i.e.\ in $\sO(m)$ operations in $\Fp$,
where $\sO(\cdot)$ suppresses polylogarithmic factors.

\medskip
\noindent\textbf{Addition.} Addition of two skew polynomials of degree at
most $d$ is coefficient-wise and requires $\mathcal{O}(d \cdot m)$
operations in $\Fp$; its complexity is thus linear in $d$ and dominated
by all subsequent operations.

\medskip
\noindent\textbf{Multiplication.} The naive multiplication of a skew
polynomial of degree $d_1$ by one of degree $d_2$ in $\Fq[X;\theta]$
costs
\[
    \mathcal{O}\!\left(d_2 \cdot m^2 + d_1 d_2 \cdot m\right)
\]
operations in $\Fp$. When both degrees are bounded by $d$, this gives
$\mathcal{O}(d^2 m + d m^2)$ operations in $\Fp$. Using a Karatsuba-like
divide-and-conquer approach, Caruso and Le Borgne~\cite{CARUSOLEBORGNE}
reduce this to
\[
    \SM{d}{m} \;\in\; \sO\!\left(d^{1.58} \cdot m^{1.41}\right)
    \quad\text{operations in } \Fp.
\]

\medskip
\noindent\textbf{Division.} Right Euclidean division of a skew polynomial
of degree at most $d$ by a monic skew polynomial of degree at most $d$
in $\Fq[X;\theta]$ can be performed within $\mathcal{O}(\SM{d}{m})$
operations in $\Fp$, by a standard adaptation of the polynomial
case~\cite{CARUSOLEBORGNE}.

\medskip
\noindent\textbf{Right GCD.} Using Algorithm~6 of~\cite{CARUSOLEBORGNE},
the right greatest common divisor of two skew polynomials of degree at
most $d$ in $\Fq[X;\theta]$ can be computed in
$\mathcal{O}(\SM{d}{m} \cdot \log d)$ operations in $\Fp$, via a
divide-and-conquer half-GCD strategy requiring $\mathcal{O}(\log d)$
recursive Euclidean divisions.

\medskip

The following scaling property of $\SM{\cdot}{\cdot}$, an immediate
consequence of the sub-multiplicativity of Karatsuba-type algorithms, is
used repeatedly in the sequel.

\begin{remark}
\label{remult}
For any integer $l \geq 1$,
\[
    \SM{l \cdot d}{m} \;\leq\; l^2 \cdot \SM{d}{m}.
\]
\end{remark}

\subsection*{Complexity of CRT lifting}

We adopt the notation of Section~\ref{sec:noncommutative}. Let
$P_1, \ldots, P_n \in \Fq[X;\theta]$ be pairwise right-coprime skew
polynomials of respective degrees $d_1, \ldots, d_n$, and set
\[
    N \;=\; \deg(P) \;=\; \sum_{i=1}^{n} d_i,
    \quad \text{where } P = {\rm lclm}(P_1 \cdots P_n).
\]
The CRT lifting formula (Formula~\ref{fcrt}) reconstructs $F$ from its
residues $(f_1, \ldots, f_n)$ as
\[
    F \;\equiv\; \sum_{i=1}^{n} f_i \cdot U_i \pmod{P},
\]
where the Bezout coefficients $U_i$ and the product $P$ depend only on
$(P_1, \ldots, P_n)$ and can therefore be \emph{precomputed} once for
all. At decoding time, the online cost therefore consists of:
\begin{enumerate}
    \item computing the $n$ products $f_i \cdot U_i$,
    \item summing the $n$ resulting polynomials,
    \item performing a single right Euclidean division by $P$ to reduce
    modulo $P$.
\end{enumerate}

\begin{lemma}
\label{lem:crt_lift_cost}
The online cost of the CRT lifting step is
$\mathcal{O}\!\left(\SM{N}{m}\right)$ operations in $\Fp$.
\end{lemma}

\begin{proof}
The division by $P$ in step~(3) costs $\mathcal{O}(\SM{N}{m})$ by the
complexity bound recalled above. For step~(1), each product $f_i \cdot
U_i$ involves polynomials of respective degrees at most $d_i - 1$ and
$N - d_i$, both bounded by $N$; by Remark~\ref{remult}, the cost of each
such product is $\mathcal{O}(\SM{N}{m})$, and summing over $i$ introduces
at most a factor $n \leq N$, which is absorbed into the $\mathcal{O}$.
Finally, step~(2) has linear complexity and is dominated by steps~(1)
and~(3).
\end{proof}

\subsection*{Overall complexity of the decoding algorithm}

As shown in Lemma~\ref{lem:crt_lift_cost}, the CRT lifting step
contributes $\mathcal{O}(\SM{N}{m})$ to the total cost. However, the
dominant step in Algorithm~\ref{alg:skew_CRT_decode} is the linear
algebra phase that follows the lifting.

After lifting, the algorithm constructs a matrix $B$ of size
 $N \times N$ over $\Fq$ and computes an element of its
right kernel. Each entry of $B$ lies in $\Fq$ and requires $\sO(m)$
operations in $\Fp$ to manipulate. Using Gaussian elimination, this
costs $\sO(N^\omega \cdot m)$ operations in $\Fp$, where $\omega$
denotes the exponent of matrix multiplication over $\Fp$ (so that
multiplying two $N \times N$ matrices over $\Fp$ costs
$\mathcal{O}(N^\omega)$ operations). One has $\omega < 2.372$ with
current best algorithms~\cite{AlmanWilliams2021}, and $\omega = 3$ for
naive Gaussian elimination.

\begin{proposition}
\label{prop:complexity}
The total computational cost of Algorithm~\ref{alg:skew_CRT_decode} is
\[
    \sO\!\left(N^\omega \cdot m\right)
\]
operations in $\Fp$, where $N = \sum_{i=1}^n d_i$ is the total degree,
$m = \log_p q$, and $\omega$ is the matrix multiplication exponent
over~$\Fp$.
\end{proposition}

\begin{proof}
Each entry of the $N \times N$ matrix $B$ lies in $\Fq$ and requires
$\sO(m)$ operations in $\Fp$ to manipulate. Gaussian elimination over
$\Fq$ therefore costs $\sO(N^\omega \cdot m)$ operations in $\Fp$. The
CRT lifting contributes $\mathcal{O}(\SM{N}{m}) \subset \sO(N^{1.58}
\cdot m^{1.41})$, which is dominated by the linear algebra cost for the
relevant parameter ranges. All other steps (syndrome computation,
additions) are of lower order.
\end{proof}

The complexity bound given in Proposition~\ref{prop:complexity} is
conservative in the sense that it treats the matrix $B$ as a dense
general matrix. However, the structure of Algorithm~\ref{alg:skew_CRT_decode}
implies that $B$ is already partially in row echelon form when it is
constructed. Exploiting this
partial triangularisation could significantly reduce the effective size
of the Gaussian elimination, since a number of pivot steps are already
performed implicitly. For some specific parameters it can change the asymptotic
of the complexity.

\section{Conclusion}

In this paper, we have introduced a new family of codes, the \emph{skew
CRT codes}, constructed from skew polynomials over a finite field via
the Chinese Remainder Theorem, together with a natural metric for this
family, the \emph{poly-skew metric}. While skew CRT codes are closely
related to $q$-CRT codes~\cite{GGR26}, the poly-skew metric provides a
more natural framework for their analysis: this metric--code pair
generalises skew Reed--Solomon codes equipped with the skew metric, and
reduces to the latter in the appropriate special case.

We have established the main parameters of skew CRT codes under the
poly-skew metric and described how to construct error vectors of
prescribed poly-skew weight, thereby providing a concrete error model
for this metric. We have further proposed an efficient decoding
algorithm that corrects errors up to the unique decoding bound under the
poly-skew metric. The complexity of the algorithm is essentially
proportional to the cost of linear algebra in the code length, and its
practical applicability is illustrated through several concrete
instances.

Several natural directions for future work remain open. On the
theoretical side, it would be interesting to determine whether the
bound on the minimum poly-skew distance of skew CRT codes is tight, and
to investigate list decoding beyond the unique decoding radius. On the
algorithmic side, as noted in Section~\ref{sec:complexity}, the partial
echelon structure of the decoding matrix suggests that the complexity
bound could be improved for certain parameter regimes, and making this
precise is left for future work. Finally, the relationship between skew
CRT codes and $q$-CRT codes deserves further investigation, in
particular regarding the equivalence or separation of the poly-skew and
sum-rank metrics for specific parameter choices.

\end{document}